\documentclass[11pt]{article}

\usepackage{paper-style}
\definecolor{citationPurple}{rgb}{0.4, 0.3, 0.7}
\definecolor{referenceGreen}{rgb}{0.15, 0.55, 0.1}

\hypersetup{
  citecolor=citationPurple,
  linkcolor=referenceGreen
}
\newtheorem*{maintheoremrestated}{\Cref{thm:main}}

\title{Optimal Shallow Circuits for Majority}
\author{%
\begin{tabular}{@{}c@{\hspace{3em}}c@{}}
Victor Lecomte & Prasanna Ramakrishnan \\
{\normalsize Alignment Research Center} &
{\normalsize Harvard University} \\
{\small\texttt{victor@alignment.org}} &
{\small\texttt{pras1712@stanford.edu}}
\end{tabular}%
}

\begin{document}

\maketitle

\begin{abstract}
Four decades on, H{\aa}stad's classical $2^{\Omega(n^{1/(d-1)})}$ lower bound for depth-$d$ circuits computing $\Parity$ remains the best known $\ACz$ circuit lower bound for any explicit function. $\Majority$ has long been a compelling candidate for stronger lower bounds: the most natural circuits computing it are substantially larger than those for $\Parity$ and have repeatedly been conjectured to be optimal.

We present a simple construction, found by GPT-6 Astra, of depth-$d$ circuits of size $2^{O(n^{1/(d - 1)})}$ for any symmetric function. This result settles the asymptotic $\ACz$ circuit complexity of $\Majority$, matching H{\aa}stad's lower bound. The proof draws inspiration from well-loved combinatorial tools, including the color-coding technique of \cite{AYZ1995}.
\end{abstract}

\section{Introduction}

At its core, complexity theory seeks to understand which functions are easy to compute, and which are hard.
Circuits offer a particularly concrete approach to probing the hardness of a function. For example, a function  expressible as a \emph{small} circuit (with few $\AND$ and $\OR$ gates) can be computed efficiently by evaluating the gates one at a time. If a function has \emph{small-depth} circuits, then this computation can be done in a highly parallelizable manner.

The study of circuit complexity is also motivated by an ambitious agenda that strikes at the heart of the most fundamental questions in computer science. If we can understand how to prove \emph{lower bounds} on the size of circuits for a given function, and eventually prove superpolynomial lower bounds for an $\NP$ function (say, $\SAT$), then we would separate $\P$ from $\NP$. Since the combinatorial structure of circuits makes them easier to reason about than Turing machines, this emerged as a promising approach.

The 1980s saw a burst of exciting progress. A sequence of results~\citep{FSS81,Ajt83,Yao85}, culminating in H{\aa}stad's \emph{switching lemma}~\citep{Has86}, established that even the $\Parity$ function  requires superpolynomial circuits of constant depth. In particular, $\Parity$ requires depth-$d$ circuits of size $2^{\Omega(n^{1/(d - 1)})}$, which is tight by a simple divide-and-conquer construction \citep{KPPY84} that splits the input into blocks, computes $\Parity$ on each block, and then computes $\Parity$ on the outputs (see \Cref{fig:parity-expansion}). That we can prove such strong bounds for such a simple function is an encouraging sign, since it should only be easier to prove lower bounds for more complex functions. 

However, that hope is yet to live up to expectations: H{\aa}stad's lower bound for $\Parity$ is still the strongest we have for \emph{any} explicit function, even for depth-3 circuits. As a signal of how far our understanding might have to come, there is still a gap even for the simple $\Majority$ function. The switching lemma establishes the same lower bound as for $\Parity$, but the natural divide-and-conquer upper bound \citep{KPPY84} carries an additional $(\log n)^{1 - 1/(d - 1)}$ factor in the exponent (see \Cref{fig:majority-expansion}). Since closing the gap could be a route to revealing the tools we need to prove better circuit lower bounds, researchers have poured substantial effort into understanding the circuit complexity of $\Majority$, even in restrictive special cases (detailed in \Cref{ssec:related}), without any asymptotic improvement to the lower bounds in the general case.

This paper gives an explanation for why the lower bounds have been so challenging to improve: they are optimal. We show that $\Majority$, and in fact any symmetric function, has depth-$d$ circuits essentially as small as those for $\Parity$.

\begin{theorem}\label{thm:main}
For any constant $d\geq 2$, every symmetric function on $n$ variables has depth-$d$ circuits of size $2^{O(n^{1/(d - 1)})}$.
\end{theorem}

Together with the lower bound of \cite{Has86}, this result settles the asymptotic small-depth circuit complexity of $\Majority$, resolving a question posed by \cite{HJP93} (on the depth-3 case).

At a very high level, the construction works in two steps:
\begin{itemize}
    \item First, design \emph{tests} certifying that the input has a Hamming weight that the function accepts. 
    \item Then, show that a randomly constructed test has sufficiently high probability of accepting a valid input, so that all inputs that should be accepted are covered by a small collection of tests. The $\OR$ of these tests is then a circuit for the function.
\end{itemize}

The meat of the proof is in the first step. The tests work by splitting the input into $k$ blocks, and comparing the Hamming weights of blocks after random shifts. If the shifts happen to align so that the shifted weights are all different modulo $k$, then we can reverse-engineer the original Hamming weight modulo $k$ in terms of the sum of shifts. We can fix the sum in advance so that a successful test certifies any particular Hamming weight we want modulo $k$, and using a few coprime moduli lets us certify the exact weight. The tests can be implemented inexpensively as small circuits (CNFs in the depth-3 case, $\Pi_{d-1}$ circuits for depth-$d$) since they only need to compare the weights of each pair of blocks, which can be broken down into computations that depend on few inputs (juntas).

\paragraph{Organization.} After covering related work and basic notation and terminology, \Cref{sec:depth-3} proves the upper bound of $2^{O(\sqrt{n})}$ for depth-3 circuits as a warm-up. That section starts with a detailed overview of the canonical divide-and-conquer construction for $\Parity$ and $\Majority$, and details one line of reasoning that suggests the stronger upper bound. Then, \Cref{sec:depth-d} extends the result to any depth by showing how the same idea can be made recursive. The two proofs follow the same template, and have some redundancy for the sake of exposition.

\paragraph{AI disclosure.} The constructions in this paper were originally found by GPT-6 Astra. The initial prompt asked to use the ideas in \cite{LRT22} to improve depth-3 circuit lower bounds for $\Majority$, which it showed was not possible. With a second prompt, Astra generalized the result to any depth. 

Our goal is to present a ``motivated explanation'' of the results geared towards human understanding (following a \href{https://terrytao.wordpress.com/2026/09/18/if-math-is-more-than-proof-we-need-to-better-celebrate-the-rest-of-it/}{recent essay} by Grant Sanderson). To that end, we reconstructed the proofs based on high-level descriptions of the ideas, instead of directly verifying Astra's write-up. The contents of the paper are human written, with additional help from Astra in copy editing, checking correctness, and producing figures.

We recognize that the $\ACz$ circuit complexity of $\Majority$ is a longstanding open problem that many researchers are fond of, and other groups may have independently reached the same results in much the same way. Rather than preparing separate manuscripts, we invite any such groups to join forces with us, and use our shared understanding to improve the explanation.

\subsection{Related work}\label{ssec:related}

Several papers have focused on developing techniques to improve circuit lower bounds for depth 3 in particular, with $\Majority$ as a central focus. Though it may seem like an awfully restrictive model, depth-3 circuits are more powerful than they look. Valiant's depth reduction~\citep{Val77,Val83} converts any linear-size, logarithmic-depth circuit with bounded fan-in into a depth-three circuit of size $2^{O(n/\log\log n)}$. Thus, a depth-three lower bound of $2^{\omega(n/\log\log n)}$ rules out linear-size circuits of logarithmic depth.

\paragraph{Restricted depth-3 circuits.} A natural approach has been to first understand more restricted constructions. Following work on circuits that combine functions of small sets of variables~\citep{HR15,GW20}, \citet{LRT22} showed that the canonical divide-and-conquer construction for $\Majority$ is optimal in a precise sense: computing $\Majority$ from functions that each depend on at most $k$ variables requires $\Omega(\frac{n}{k}\log k)$ such functions, even allowing their inputs to overlap arbitrarily, for $k\leq n^{1-\varepsilon}$ with fixed $\varepsilon>0$. Another line of work limits the fan-in of the bottom gates. Here, a central question is how many inputs of weight $t$ a $k$-CNF can accept while rejecting every lighter input: an upper bound on this number gives a lower bound on how many CNFs are needed to compute the corresponding threshold function. \citet{GKP24} resolve this question for $k=2$ and propose bounds for larger $k$ that would imply a $2^{\Omega(\sqrt{n\log n})}$ depth-3 lower bound for $\Majority$ (ruled out by \Cref{thm:main}). For bottom fan-in three, \citet{GPPST24} proved a $1.251^n$ lower bound for $\Majority$, and \citet{GKPRST26} obtained an optimal bound in the monotone version of this model. 

These results identify obstacles that any improved construction must take into account. Our depth-3 construction still uses local computations, but organizes their answers into CNF tests that can be combined efficiently with just an $\OR$. It also uses non-monotone circuits and allows the bottom fan-in to grow with $n$. 

The search for an explicit function requiring depth-3 circuits of size $2^{\omega(\sqrt n)}$ remains open, and \Cref{thm:main} suggests that it may require more complicated functions.

\paragraph{Top-down lower bounds.} The lower bounds derived via the switching lemma are often thought of as ``bottom-up'': starting at the input layer, they argue that a random restriction which fixes most variables often retains much of the complexity of the original function (e.g., a restriction of $\Parity$ is still a $\Parity$), but simplifies a small circuit enough to reduce its depth (in particular, turning the CNFs at the bottom two layers into DNFs, or vice versa).
\cite{HJP93} showed that a completely different ``top-down'' approach, reasoning about the subcircuits just below the output gate, can establish comparable lower bounds. 
They proved lower bounds of $2^{0.618\sqrt n}$ for $\Parity$ and $2^{0.849\sqrt n}$ for $\Majority$. The top-down viewpoint has been further developed using information-theoretic methods~\citep{MW19} and extended to depth-4 circuits~\citep{GRSS23,WL26}, and very recently to depth-$d$ for any constant $d$~\citep{KardesRossman2026}.

\paragraph{Top-down upper bounds.} In a similar spirit, the proof of \Cref{thm:main} can be viewed as a top-down \emph{upper bound}. For depth 3, we view the circuit as an $\OR$ of CNFs, and each CNF is a test used to certify that the input has a Hamming weight that should be accepted.\footnote{Along similar lines, \citet{Hirahara2017} proves a general connection between one-sided CNF tests like these and depth-3 formulas.}
Top-down approaches like these are natural ways to construct small-depth circuits, and we mention a couple of papers that share similarities with our construction. 

\citet{Amano2023} uses computer search to find CNFs on a constant number of variables that certify $\Majority$ on many inputs. He combines these on disjoint blocks and uses a probabilistic argument to cover every valid input (akin to the last step in our proof). His construction improves upper bounds for $\Majority$ in depth-3 circuits with fixed small bounds on the number of positive literals in each bottom gate.
We also note that an earlier paper of \citet{Wol06} claimed a $2^{O(\sqrt n)}$ depth-3 upper bound for symmetric functions, but its proof has a gap (discussed by \citet{LRT22} and \citet{GKP24}). The proposed construction similarly combines small CNF tests through a probabilistic covering argument, but uses a different family of tests.

\section{Preliminaries}
\label{sec:preliminaries}

\paragraph{Functions.} A Boolean function on $n$ variables maps the hypercube $\{0, 1\}^n$ to $\{0, 1\}$. For $x \in \{0, 1\}^n$, $|x|$ denotes its \emph{Hamming weight} (or simply \emph{weight}), the number of coordinates of $x$ set to $1$.

A Boolean function is \emph{symmetric} if its output depends only on the weight of its input. For example, the $\Majority$ function is $\textbf{1}[|x| \geq n/2]$, and the $\Parity$ function is $|x| \pmod{2}$. In general, any symmetric function can be written as
\begin{equation}\label{eq:decomp}
 \bigvee_{t \in S} \textbf{1}[|x| = t]
\end{equation}
for a set $S \subseteq \{0, 1, \dots, n\}$. This decomposition allows us to focus on constructing circuits for $\textbf{1}[|x| = t]$ for fixed $t$, often denoted $\mathsf{EXACT}_t(x)$ in the literature.

Finally, a \emph{junta} is a Boolean function that depends on few variables. In particular, a \emph{$k$-junta} depends on at most $k$ variables.

\paragraph{Circuits.} A circuit is a directed acyclic graph with input nodes labeled by literals $x_i$ or $\neg x_i$, along with $\AND$ ($\wedge$) and $\OR$ ($\vee$) gates of arbitrary fan-in. We view circuits as alternating layers of $\AND$ and $\OR$ gates, with a single output gate at the top. The \emph{depth} is the number of gate layers, and the \emph{size} is the number of $\AND$ and $\OR$ gates. A depth-$d$ circuit whose output gate is an $\OR$ (resp. $\AND$) is called a $\Sigma_d$ (resp. $\Pi_d$) circuit  (see \Cref{fig:circuit-example}).\footnote{The notation comes from the convention of viewing $\OR$ as Boolean sum and $\AND$ as Boolean product.} We always treat $d$ as a fixed constant, so constants hidden by big-$O$ notation may depend on $d$. 

\begin{figure}[htbp]
\centering
\input{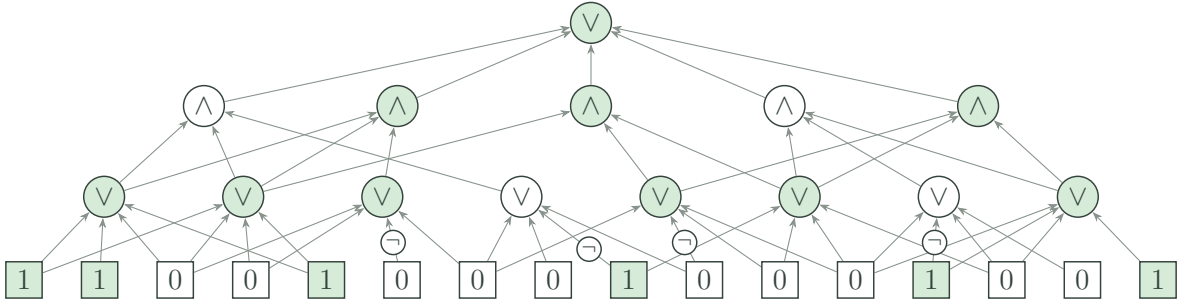}
\caption{A $\Sigma_3$ circuit on $n = 16$ input bits. Green (resp. white) inputs and gates have value $1$ (resp. 0). The edges are \emph{wires} carrying values upward. For convenience, a circled $\lnot$ negates the value on that wire, but in reality we would have separate input nodes for each $x_i$ and $\lnot x_i$.}
\label{fig:circuit-example}
\end{figure}

By De Morgan's laws, the negation of a $\Sigma_d$ circuit is a $\Pi_d$ circuit of the same size, and vice versa. Since the negation of a symmetric function is also symmetric, it suffices to prove our upper bounds for one of these two kinds of circuits. We focus on $\Sigma_d$ circuits, whose output $\OR$ lets us combine the circuits in the decomposition by weight in \Cref{eq:decomp} without increasing the depth.

A $\Pi_2$ circuit is a CNF (conjunctive normal form), and a $\Sigma_2$ circuit is a DNF (disjunctive normal form). The size of such a circuit is its number of clauses (for a CNF) or terms (for a DNF), plus one for the output gate. For convenience, we call the number of clauses or terms its size.

A useful fact is that every Boolean function on $n$ variables can be represented as both a CNF and a DNF of size at most $2^n$. For a DNF, take the $\OR$ over all satisfying assignments, with each assignment represented by an $\AND$ of the literals specifying it. The CNF representation follows by applying this construction to the negation of the function and using De Morgan's laws. We repeatedly use the helpful corollary that every $k$-junta can be represented as both a CNF and a DNF of size at most $2^k$.

\section{Depth 3}
\label{sec:depth-3}

Since deeper circuits are harder to build intuition for, we will warm up with the simplest nontrivial case: depth 3. To motivate the construction, we will start by explaining the canonical constructions of small depth-3 circuits for $\Parity$ and $\Majority$, and the change in perspective that suggests a stronger construction. This explanation is intentionally informal, taking liberties with rounding, inconsequential constants, and lower order terms.

The typical approach is to divide and conquer. For example, we can compute the $\Parity$ function on $n$ bits by splitting the input into $\sqrt{n}$ blocks of size $\sqrt{n}$, computing $\Parity$ on each block, and then computing $\Parity$ on the outputs. Each $\Parity$ function is a $\sqrt{n}$-junta, which can expand into a CNF or DNF of size $2^{\sqrt{n}}$ (see \Cref{fig:parity-expansion}). Expanding the outer $\Parity$ into a DNF and each inner $\Parity$ into a CNF, we can collapse the middle $\AND$ layers to get a depth-3 ($\Sigma_3$) circuit of size $2^{O(\sqrt{n})}$.\footnote{We gloss over a slight subtlety: the inner $\Parity$ functions which become inputs to the outer $\Parity$ may need to be negated, preventing us from collapsing the layers. To address this, we really replace each inner $\Parity$ with \emph{two} CNFs, one computing $\Parity$ and another for its negation.} 

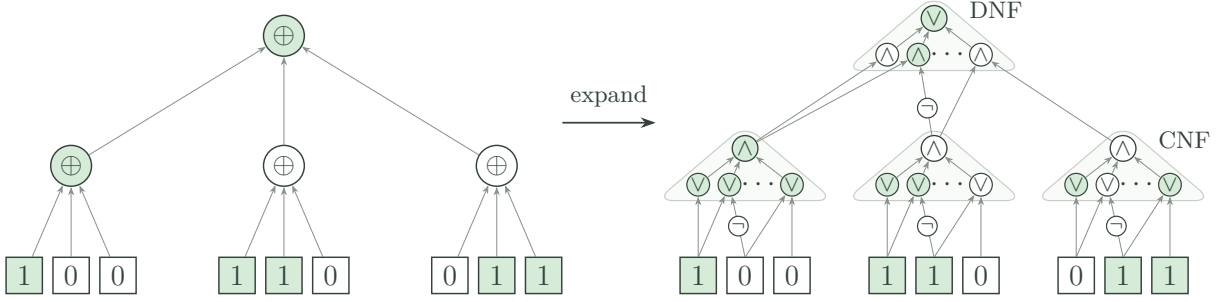
\begin{figure}[htbp]
\centering
\begingroup%
\def\parityBlocks{{{1,0,0},{1,1,0},{0,1,1}}}%
\def\parityEvenAssignments{{{0,0,0},{0,1,1},{1,0,1},{1,1,0}}}%
\def\parityOddAssignments{{{0,1,0},{1,0,0},{0,0,1},{1,1,1}}}%
\def\parityShownEdges{1/1,1/2,2/2,2/4,3/4}%
\definecolor{parityGreen}{HTML}{D6EBD8}%
\definecolor{parityInk}{HTML}{34423C}%
\definecolor{parityWire}{HTML}{87928B}%
\pgfmathsetlengthmacro{\parityUnit}{\linewidth/160}%
\def\parityHeightDivisor{4}%
\pgfmathsetlengthmacro{\parityVerticalUnit}{\linewidth/\parityHeightDivisor/40}%
\def\parityModuleShift{2.3}%
\def\parityOutlineShift{0.4}%
\begin{tikzpicture}[
  x=\parityUnit,
  y=\parityVerticalUnit,
  font=\fontsize{11}{13}\selectfont,
  parity/node/.style={draw=parityInk, text=parityInk, line width=0.55pt, inner sep=0pt, outer sep=0pt, fill=white},
  parity/input/.style={parity/node, rectangle, minimum size=4.8mm},
  parity/gate/.style={parity/node, circle, minimum size=5.4mm},
  parity/right gate/.style={parity/gate, minimum size=3.4mm, line width=0.5pt, font=\fontsize{9}{10.5}\selectfont},
  parity/small gate/.style={parity/gate, minimum size=3mm, line width=0.45pt, font=\fontsize{8.5}{10}\selectfont},
  parity/inversion/.style={parity/node, circle, minimum size=2.6mm, line width=0.45pt, font=\fontsize{7}{8}\selectfont},
  parity/state/.is choice,
  parity/state/0/.style={fill=white},
  parity/state/1/.style={fill=parityGreen},
  parity/wire/.style={draw=parityWire, line width=0.35pt, -{Stealth[length=1.1mm,width=0.8mm]}, shorten <=0.3pt, shorten >=0.3pt},
  parity/triangle/.style={draw=parityWire!45, fill=parityGreen!18!white, line width=0.45pt, rounded corners=2.2mm, line join=round},
  parity/label/.style={text=parityInk, font=\fontsize{9}{11}\selectfont, inner sep=0pt},
  parity/ellipsis/.style={text=parityInk, font=\fontsize{11}{13}\selectfont, inner sep=0pt},
  declare function={
    paritybit(\b,\i)=\parityBlocks[\b-1][\i-1];
    blockparity(\b)=mod(paritybit(\b,1)+paritybit(\b,2)+paritybit(\b,3),2);
    parityclause(\b,\j)=max(abs(paritybit(\b,1)-\parityEvenAssignments[\j-1][0]),abs(paritybit(\b,2)-\parityEvenAssignments[\j-1][1]),abs(paritybit(\b,3)-\parityEvenAssignments[\j-1][2]));
    parityterm(\j)=(1-abs(blockparity(1)-\parityOddAssignments[\j-1][0]))*(1-abs(blockparity(2)-\parityOddAssignments[\j-1][1]))*(1-abs(blockparity(3)-\parityOddAssignments[\j-1][2]));
  }
]
  \path[use as bounding box] (0,0) rectangle (160,40);

  \foreach \b in {1,2,3} {
    \pgfmathsetmacro{\leftCenter}{8.85+28.15*(\b-1)}
    \pgfmathsetmacro{\rightCenter}{98+25*(\b-1)}
    \foreach \i in {1,2,3} {
      \pgfmathtruncatemacro{\value}{paritybit(\b,\i)}
      \node[parity/input, parity/state=\value] (left-input-\b-\i) at ({\leftCenter+6.2*(\i-2)},2.8) {$\value$};
      \node[parity/input, parity/state=\value] (right-input-\b-\i) at ({\rightCenter+6.2*(\i-2)},2.8) {$\value$};
    }
    \pgfmathtruncatemacro{\value}{blockparity(\b)}
    \node[parity/gate, parity/state=\value] (left-parity-\b) at (\leftCenter,17.3) {$\oplus$};
    \node[parity/right gate, parity/state=\value] (cnf-output-\b) at (\rightCenter,{17.3+\parityModuleShift}) {$\wedge$};
    \foreach \j in {1,2,3,4} {
      \pgfmathtruncatemacro{\value}{parityclause(\b,\j)}
      \ifnum\j=3
        \node[parity/ellipsis] at ({\rightCenter+(12.4/3)*(\j-2.5)},{12.5+\parityModuleShift}) {$\cdots$};
      \else
        \node[parity/small gate, parity/state=\value] (cnf-clause-\b-\j) at ({\rightCenter+(12.4/3)*(\j-2.5)},{12.5+\parityModuleShift}) {$\vee$};
      \fi
    }
  }

  \pgfmathtruncatemacro{\parityOutput}{mod(blockparity(1)+blockparity(2)+blockparity(3),2)}
  \node[parity/gate, parity/state=\parityOutput] (left-output) at (37,34.5) {$\oplus$};
  \node[parity/right gate, parity/state=\parityOutput] (right-output) at (123,{34.5+\parityModuleShift}) {$\vee$};
  \foreach \j in {1,2,3,4} {
    \pgfmathtruncatemacro{\value}{parityterm(\j)}
    \ifnum\j=3
      \node[parity/ellipsis] at ({123+(12.4/3)*(\j-2.5)},{29.7+\parityModuleShift}) {$\cdots$};
    \else
      \node[parity/small gate, parity/state=\value] (dnf-term-\j) at ({123+(12.4/3)*(\j-2.5)},{29.7+\parityModuleShift}) {$\wedge$};
    \fi
  }

  \node[parity/label, anchor=west] at (127.6,{35.7+\parityModuleShift}) {DNF};
  \node[parity/label, anchor=west] at (152.6,{18.5+\parityModuleShift}) {CNF};

  \node[parity/label] (parity-expand-label) at (80,26.5) {expand};
  \draw[draw=parityInk, line width=0.65pt, -{Stealth[length=2mm,width=1.4mm]}, shorten <=-1mm, shorten >=-1mm] (parity-expand-label.west |- 0,23) -- (parity-expand-label.east |- 0,23);

  \begin{scope}[on background layer]
    \path[parity/triangle] (111.25,{27.2+\parityModuleShift+\parityOutlineShift}) -- (134.75,{27.2+\parityModuleShift+\parityOutlineShift}) -- (123,{37.2+\parityModuleShift+\parityOutlineShift}) -- cycle;
    \foreach \b in {1,2,3} {
      \pgfmathsetmacro{\rightCenter}{98+25*(\b-1)}
      \path[parity/triangle] ({\rightCenter-11.75},{10+\parityModuleShift+\parityOutlineShift}) -- ({\rightCenter+11.75},{10+\parityModuleShift+\parityOutlineShift}) -- (\rightCenter,{20+\parityModuleShift+\parityOutlineShift}) -- cycle;
    }
    \foreach \b in {1,2,3} {
      \foreach \i in {1,2,3} {
        \draw[parity/wire] (left-input-\b-\i) -- (left-parity-\b);
      }
      \foreach \i/\j in \parityShownEdges {
        \draw[parity/wire] (right-input-\b-\i.north) -- (cnf-clause-\b-\j);
      }
      \draw[parity/wire] (left-parity-\b) -- (left-output);
      \foreach \j in {1,2,4} {
        \draw[parity/wire] (cnf-clause-\b-\j) -- (cnf-output-\b);
      }
    }
    \foreach \b/\j in \parityShownEdges {
      \draw[parity/wire] (cnf-output-\b) -- (dnf-term-\j);
    }
    \foreach \j in {1,2,4} {
      \draw[parity/wire] (dnf-term-\j) -- (right-output);
    }
  \end{scope}
  \foreach \b in {1,2,3} {
    \path (right-input-\b-2.north) -- (cnf-clause-\b-2.center) node[pos=0.43, parity/inversion] {$\lnot$};
  }
  \path (cnf-output-2.center) -- (dnf-term-2.center) node[pos=0.43, parity/inversion] {$\lnot$};
\end{tikzpicture}%
\endgroup%
\caption{Computing $\Parity$ with divide-and-conquer, and expanding into a small $\Sigma_3$ circuit (after collapsing the adjacent $\AND$ layers). We use $\oplus$ to label a parity gate.}
\label{fig:parity-expansion}
\end{figure}

Almost the same approach works for $\Majority$. Again, we can split the input into $k$ blocks of size $n/k$. However, we cannot simply compute $\mathsf{Majority}$ as a majority of majorities.\footnote{This pesky fact leads to a phenomenon called \emph{gerrymandering}.} Instead, we can compute the Hamming weight of each block, and
check whether the sum of the weights is at least $n/2$. Each bit of a block's weight is an $n/k$-junta, and since each block's weight uses $\log(n/k)$ bits, the outer function depends on $k \log(n/k)$ such bits. Expanding these juntas as before, the outer function contributes $2^{k \log (n/k)}$ gates and each inner junta contributes $2^{O(n/k)}$ gates. Parameters balance at $k \approx \sqrt{n/\log n}$, giving a depth-3 circuit of size $2^{O(\sqrt{n \log n})}$.

If we believe that divide-and-conquer is the right approach, then this construction is abundantly natural.\footnote{Indeed, \citet{LRT22} show that it is essentially the best divide-and-conquer construction.} 
However, we can view the construction in a very different way if we look at the final depth-3 circuit we get after expansion as an $\OR$ of CNFs.  
We can think of each CNF as a \emph{test} certifying that the input satisfies $\Majority$. With some effort, one can decode the form of these tests. Each test sets a threshold for each block, and checks that each block's weight is at least its threshold.
 The sum of thresholds is at least $n/2$, so an input passes a test only if it satisfies $\Majority$. Taking all such tests covers every valid input, and there are roughly $(n/k)^k = 2^{k\log(n/k)}$ of them. Each test is a CNF of size roughly $2^{n/k}$ since it is an $\AND$ over $n/k$-juntas, and so we get the same size bound as before after setting $k \approx \sqrt{n/\log n}$.

\begin{figure}[htbp]
\centering
\begingroup%
\def\majorityBits{{0,1,0,0,1,1,0,0,1,1}}%
\def\majorityClauseInputs{{{-1,2,3},{1,2,3},{1,2,-3},{2,3,0},{4,5,6},{4,5,-6},{7,8,9},{7,0,0},{7,8,-9}}}%
\def\majorityTests{{{3,2,0},{2,2,1},{1,2,2},{0,2,3}}}%
\def\majorityClausePositions{{91.8,99,106.2,115,125,133,140.5,147,157.5}}%
\def\majorityTestPositions{{99,117,135,154}}%
\def\majorityTestEdges{1/1,2/1,3/1,4/2,5/2,6/2,5/3,6/3,7/3,9/3,6/4,8/4,9/4}%
\definecolor{majorityGreen}{HTML}{D6EBD8}%
\definecolor{majorityInk}{HTML}{34423C}%
\definecolor{majorityWire}{HTML}{87928B}%
\pgfmathsetlengthmacro{\majorityUnit}{\linewidth/160}%
\def\majorityHeightDivisor{4}%
\pgfmathsetlengthmacro{\majorityVerticalUnit}{\linewidth/\majorityHeightDivisor/40}%
\pgfmathsetmacro{\majorityLayerGap}{(37.3-2.8-11.9mm/\majorityVerticalUnit)/3}%
\pgfmathsetmacro{\majorityClauseY}{2.8+4.2mm/\majorityVerticalUnit+\majorityLayerGap}%
\pgfmathsetmacro{\majorityTestY}{\majorityClauseY+3.7mm/\majorityVerticalUnit+\majorityLayerGap}%
\begin{tikzpicture}[
  x=\majorityUnit,
  y=\majorityVerticalUnit,
  font=\fontsize{11}{13}\selectfont,
  majority/node/.style={draw=majorityInk, text=majorityInk, line width=0.55pt, inner sep=0pt, outer sep=0pt, fill=white},
  majority/input/.style={majority/node, rectangle, minimum size=4.8mm},
  majority/gate/.style={majority/node, circle, minimum size=3.8mm, line width=0.5pt, font=\fontsize{9.5}{11}\selectfont},
  majority/clause/.style={majority/gate, minimum size=3.6mm},
  majority/output/.style={majority/gate, minimum size=4.2mm},
  majority/inversion/.style={majority/node, circle, minimum size=2.6mm, line width=0.45pt, font=\fontsize{7}{8}\selectfont},
  majority/weight box/.style={majority/node, rectangle, rounded corners=1.5mm, minimum width=13.4mm, minimum height=10.8mm},
  majority/sum box/.style={majority/weight box, minimum width=17mm},
  majority/state/.is choice,
  majority/state/0/.style={fill=white},
  majority/state/1/.style={fill=majorityGreen},
  majority/wire/.style={draw=majorityWire, line width=0.35pt, -{Stealth[length=1.1mm,width=0.8mm]}, shorten <=0.3pt, shorten >=0.3pt},
  majority/triangle/.style={draw=majorityWire!45, fill=majorityGreen!18!white, line width=0.45pt, rounded corners=2.2mm, line join=round},
  majority/label/.style={text=majorityInk, font=\fontsize{9}{11}\selectfont, inner sep=0pt},
  majority/module label/.style={majority/label, text=majorityInk!72!white, font=\fontsize{8}{9.5}\selectfont},
  majority/ellipsis/.style={text=majorityInk, font=\fontsize{11}{13}\selectfont, inner sep=0pt},
  declare function={
    majoritybit(\i)=\majorityBits[\i];
    majorityliteral(\i)=abs(majoritybit(abs(\i))-(\i<0));
    majorityweight(\b)=majoritybit(3*\b-2)+majoritybit(3*\b-1)+majoritybit(3*\b);
    majorityclause(\j)=max(majorityliteral(\majorityClauseInputs[\j-1][0]),majorityliteral(\majorityClauseInputs[\j-1][1]),majorityliteral(\majorityClauseInputs[\j-1][2]));
    majoritytest(\j)=(majorityweight(1)>=\majorityTests[\j-1][0])*(majorityweight(2)>=\majorityTests[\j-1][1])*(majorityweight(3)>=\majorityTests[\j-1][2]);
  }
]
  \path[use as bounding box] (0,0) rectangle (160,40);

  \foreach \b in {1,2,3} {
    \pgfmathsetmacro{\leftCenter}{8.85+28.15*(\b-1)}
    \pgfmathsetmacro{\rightCenter}{98+25*(\b-1)}
    \foreach \i in {1,2,3} {
      \pgfmathtruncatemacro{\index}{3*(\b-1)+\i}
      \pgfmathtruncatemacro{\value}{majoritybit(\index)}
      \node[majority/input, majority/state=\value] (majority-left-input-\index) at ({\leftCenter+6.2*(\i-2)},2.8) {$\value$};
      \node[majority/input, majority/state=\value] (majority-right-input-\index) at ({\rightCenter+6.2*(\i-2)},2.8) {$\value$};
    }
    \node[majority/weight box] (majority-weight-\b) at (\leftCenter,17.2) {};
    \foreach \j in {1,2} {
      \pgfmathtruncatemacro{\value}{mod(floor(majorityweight(\b)/pow(2,2-\j)),2)}
      \node[majority/gate, majority/state=\value] (majority-weight-bit-\b-\j) at ({\leftCenter+5.2*(\j-1.5)},19.3) {$\value$};
    }
    \node[majority/module label] at (\leftCenter,14.8) {weight};
  }

  \pgfmathtruncatemacro{\majorityOutput}{majorityweight(1)+majorityweight(2)+majorityweight(3)>=5}
  \node[majority/sum box] (majority-sum) at (37,34.1) {};
  \node[majority/gate, majority/state=\majorityOutput] at (37,36.5) {$\majorityOutput$};
  \node[majority/module label] at (37,31.65) {$\mathrm{sum}\geq n/2$};

  \foreach \j in {1,...,9} {
    \pgfmathtruncatemacro{\value}{majorityclause(\j)}
    \node[majority/clause, majority/state=\value] (majority-clause-\j) at ({\majorityClausePositions[\j-1]},\majorityClauseY) {$\vee$};
  }
  \foreach \j in {1,...,4} {
    \pgfmathtruncatemacro{\value}{majoritytest(\j)}
    \node[majority/gate, majority/state=\value] (majority-test-\j) at ({\majorityTestPositions[\j-1]},\majorityTestY) {$\wedge$};
  }
  \node[majority/output, majority/state=\majorityOutput] (majority-output) at (123,37.3) {$\vee$};
  \node[majority/ellipsis] at (152,\majorityClauseY) {$\cdots$};
  \node[majority/ellipsis] at (144.5,\majorityTestY) {$\cdots$};
  \node[majority/label] at (99,{\majorityTestY+6.7}) {test CNF};

  \node[majority/label] (majority-expand-label) at (80,26.5) {expand};
  \draw[draw=majorityInk, line width=0.65pt, -{Stealth[length=2mm,width=1.4mm]}, shorten <=-1mm, shorten >=-1mm] (majority-expand-label.west |- 0,23) -- (majority-expand-label.east |- 0,23);

  \begin{scope}[on background layer]
    \path[majority/triangle] (85.5,{\majorityClauseY-3.4}) -- (112.5,{\majorityClauseY-3.4}) -- (99,{\majorityTestY+4.6}) -- cycle;
    \foreach \b in {1,2,3} {
      \foreach \i in {1,2,3} {
        \pgfmathtruncatemacro{\index}{3*(\b-1)+\i}
        \draw[majority/wire] (majority-left-input-\index.north) -- (majority-weight-\b.south);
      }
    }
    \foreach \j in {1,...,9} {
      \foreach \i in {0,1,2} {
        \pgfmathtruncatemacro{\source}{abs(\majorityClauseInputs[\j-1][\i])}
        \ifnum\source>0
          \pgfmathtruncatemacro{\showWire}{!((\j==4 && \source==2)||(\j==3 && \source==1)||(\j==1 && \source==2)||(\j==1 && \source==3))}
          \ifnum\showWire=1
            \draw[majority/wire] (majority-right-input-\source.north) -- (majority-clause-\j.south);
          \fi
        \fi
      }
    }
    \foreach \clause/\test in \majorityTestEdges {
      \draw[majority/wire] (majority-clause-\clause.north) -- (majority-test-\test.south);
    }
    \foreach \j in {1,...,4} {
      \draw[majority/wire] (majority-test-\j.north) -- (majority-output.south);
    }
  \end{scope}
  \foreach \b in {1,2,3} {
    \foreach \j in {1,2} {
      \draw[majority/wire] (majority-weight-bit-\b-\j.north) -- (majority-sum.south);
    }
  }
  \foreach \j in {1,...,9} {
    \foreach \i in {0,1,2} {
      \pgfmathtruncatemacro{\literal}{\majorityClauseInputs[\j-1][\i]}
      \ifnum\literal<0
        \pgfmathtruncatemacro{\source}{abs(\literal)}
        \path (majority-right-input-\source.north) -- (majority-clause-\j.south) node[pos=0.55, majority/inversion] {$\lnot$};
      \fi
    }
  }
\end{tikzpicture}%
\endgroup%
\caption{Computing $\Majority$ with divide-and-conquer, and expanding into an $\OR$ of CNF tests.}
\label{fig:majority-expansion}
\end{figure}
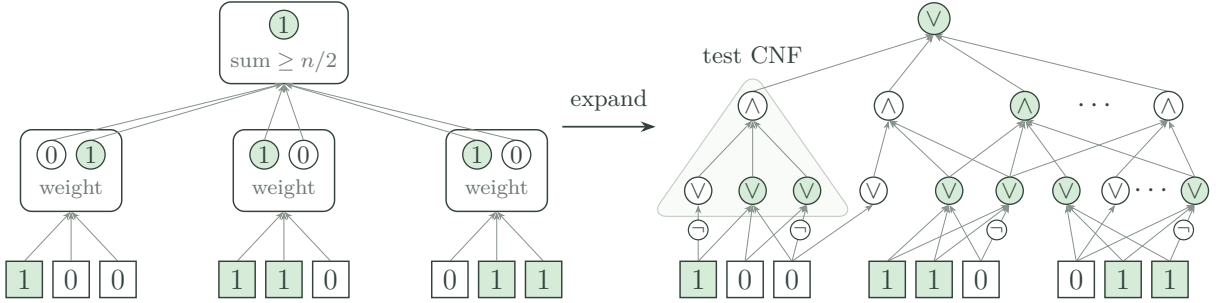

Our construction follows the same template: design tests that are expressible as small CNFs ($\AND$s of juntas), then ensure every valid input is covered with a small collection of tests. 
To see where we might improve, consider how the previous tests cover inputs of weight exactly $n/2$. 
If such an input passes a test, every block must meet its threshold \emph{exactly}: if one block exceeds the threshold, another falls short. While the test does certify that the total weight is valid, it only does so for a thin slice of inputs that manage to thread the needle. This restrictiveness forces us to have many tests. The core challenge is in designing tests that certify the Hamming weight, while accepting many valid inputs.

Let us focus on certifying that the input $x$ has some particular weight $t$. This is enough for any symmetric function, since we can take the $\OR$ over the weights it accepts. We can simplify the goal further by first certifying that $|x| \equiv t \pmod{m}$ for some modulus $m$. By the Chinese Remainder Theorem, congruences for two coprime moduli at least $\sqrt{n}$ determine the weight exactly, and we can combine their CNFs with an $\AND$ without increasing the depth or affecting the size asymptotically.

To certify the total weight modulo $m$ without fixing every block's weight, we have to go beyond checking one block at a time.
A condition on a \emph{few} blocks still depends on few input variables, so it also has a small CNF. We can therefore look to impose local constraints between the weights of different blocks, and hope that these constraints certify some particular weight.

A simple idea does the trick: set $m$ to be exactly $k$ (the number of blocks), and take the constraint that the block weights are \emph{distinct} modulo $k$. Distinctness can be checked by comparing each pair of blocks, so the constraint is an $\AND$ of juntas --- a small CNF. The sum of weights modulo $k$ is then certified to be $0 + 1 + \cdots + (k - 1) \equiv \binom{k}{2} \pmod{k}$, without prescribing which block has which weight. Heuristically, this CNF accepts about a $\frac{k!}{k^k} \approx e^{-k}$ fraction of inputs, whereas a CNF prescribing a weight for each block only accepts a $k^{-k}$ fraction of inputs. We are setting $k = m \approx \sqrt{n}$, so such a CNF will accept an $e^{-\sqrt{n}}$ fraction of valid inputs, and we will therefore only need about $2^{O(\sqrt n)}$ of them.

To extend to certifying any weight $t$ modulo $k$, not just $\binom{k}{2}$, the symmetry of the cyclic group suggests an easy fix. Instead of working with the block weights directly, apply the distinctness constraints after \emph{shifting} each weight by fixed amounts. If we set the total shifts to be $\binom{k}{2} - t \pmod{k}$ the distinctness constraint precisely certifies $|x| \equiv t \pmod{k}$. From there, we are home free.

With this overview in mind, we now present the result formally.

\begin{theorem}
Every symmetric function on $n$ variables has depth-$3$ circuits of size $2^{O(\sqrt{n})}$. 
\end{theorem}

\begin{proof}
Since every symmetric function is the $\OR$ of at most $n + 1$ functions of the form $\textbf{1}[|x| = t]$, it suffices to construct a $\Sigma_3$ circuit ($\OR\circ\AND\circ\OR$) of size $2^{O(\sqrt{n})}$ for each such function.

The proof has two stages: first, design test CNFs of size $2^{O(\sqrt{n})}$ that certify that $x$ has weight $t$, then show that some collection of $2^{O(\sqrt{n})}$ tests covers every input of weight $t$.

\paragraph{Anatomy of a test.} As a starting point, we will construct tests that certify that $|x| \equiv t \pmod{k}$ for any fixed $k$.

Consider the following simple test. Split the input $x$ into $k$ blocks of length at most $\ceil{n/k}$, and let $w_i$ denote the weight in block $i$. Fix shifts $s_1, \dots, s_k \in \{0,1,\dots,k-1\}$, one for each block, and check whether $w_1 + s_1, \dots, w_k + s_k$ are all distinct modulo $k$. If they are, it follows that 
$$\sum_{i=1}^k (w_i + s_i) \equiv 0 + 1 + \cdots + (k-1) \equiv \binom{k}{2} \pmod{k}$$
and so 
$$\sum_{i=1}^k w_i \equiv \binom{k}{2} - \sum_{i=1}^k s_i \pmod{k}.$$
Thus, if we choose the shifts so that 
\begin{equation}\label{eq:shift-constraint}
\sum_{i=1}^k s_i \equiv \binom{k}{2} - t \pmod{k},
\end{equation}
then passing the test certifies that $|x| \equiv t \pmod{k}$.

Even though the test queries every input variable, it can be implemented as a small CNF. The reason is simply that checking the distinctness of $w_1 + s_1, \dots, w_k + s_k$ can be done two at a time, and each pair comparison depends on far fewer variables. In particular a test computes 
\begin{equation}\label{eq:test}
\textbf{1}[w_1 + s_1,\dots, w_k + s_k \text{ are distinct modulo } k] = \bigwedge_{1\le i < j \le k} \textbf{1}[w_i + s_i \not\equiv w_j + s_j \pmod k].
\end{equation} 
Each constituent function $\textbf{1}[w_i + s_i \not\equiv w_j + s_j \pmod k]$ is a $2 \ceil{n/k}$-junta, so it can be written as a CNF of size $2^{2 \ceil{n/k}}$ (see \Cref{fig:pairwise-test}). Collapsing the outer $\AND$ of this CNF with the outer $\AND$ in \Cref{eq:test}, the overall test is a CNF of size at most $\binom{k}{2} 2^{2 \ceil{n/k}}$. 

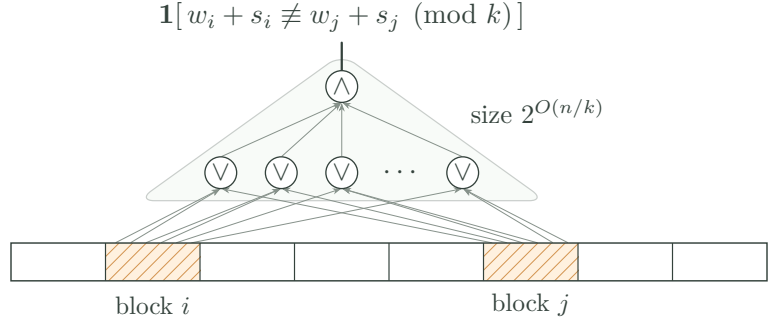
\begin{figure}[htbp]
\centering
\begingroup%
\def\pairTestBlockCount{8}%
\def\pairTestFirstBlock{2}%
\def\pairTestSecondBlock{6}%
\pgfmathsetmacro{\pairTestBlockWidth}{100/\pairTestBlockCount}%
\pgfmathsetmacro{\pairTestFirstLeft}{8+(\pairTestFirstBlock-1)*\pairTestBlockWidth}%
\pgfmathsetmacro{\pairTestSecondLeft}{8+(\pairTestSecondBlock-1)*\pairTestBlockWidth}%
\pgfmathsetmacro{\pairTestCenter}{(\pairTestFirstLeft+\pairTestSecondLeft+\pairTestBlockWidth)/2}%
\pgfmathsetlengthmacro{\pairTestUnit}{\linewidth/160}%
\definecolor{pairTestInk}{HTML}{34423C}%
\definecolor{pairTestWire}{HTML}{87928B}%
\definecolor{pairTestGreen}{HTML}{D6EBD8}%
\definecolor{pairTestOrange}{HTML}{D39450}%
\definecolor{pairTestOrangeTint}{HTML}{FFF1DF}%
\begin{tikzpicture}[
  x=\pairTestUnit,
  y=\pairTestUnit,
  font=\fontsize{11}{13}\selectfont,
  pair test/gate/.style={circle, draw=pairTestInk, text=pairTestInk, fill=white, line width=0.5pt, inner sep=0pt, outer sep=0pt, minimum size=4.2mm, font=\fontsize{10.5}{12}\selectfont},
  pair test/wire/.style={draw=pairTestWire, line width=0.35pt, -{Stealth[length=1.1mm,width=0.8mm]}, shorten <=0.3pt, shorten >=0.3pt},
  pair test/output wire/.style={draw=pairTestInk, line width=0.85pt, line cap=round},
  pair test/triangle/.style={draw=pairTestWire!45, fill=pairTestGreen!18!white, line width=0.45pt, rounded corners=2.2mm, line join=round},
  pair test/label/.style={text=pairTestInk, font=\fontsize{9.5}{11}\selectfont, inner sep=0pt},
  pair test/ellipsis/.style={text=pairTestInk, font=\fontsize{11}{13}\selectfont, inner sep=0pt}
]
  \path[use as bounding box] (0,0) rectangle (116,44);

  \path[fill=white] (8,6) rectangle (108,11);
  \foreach \block in {\pairTestFirstBlock,\pairTestSecondBlock} {
    \pgfmathsetmacro{\blockLeft}{8+(\block-1)*\pairTestBlockWidth}
    \path[fill=pairTestOrangeTint] (\blockLeft,6) rectangle ({\blockLeft+\pairTestBlockWidth},11);
    \begin{scope}
      \clip (\blockLeft,6) rectangle ({\blockLeft+\pairTestBlockWidth},11);
      \pgfmathtruncatemacro{\hatchCount}{ceil((\pairTestBlockWidth+5)/1.8)}
      \foreach \hatch in {0,...,\hatchCount} {
        \draw[draw=pairTestOrange, line width=0.3pt] ({\blockLeft-5+1.8*\hatch},6) -- ({\blockLeft+1.8*\hatch},11);
      }
    \end{scope}
  }
  \draw[draw=pairTestInk, line width=0.55pt] (8,6) rectangle (108,11);
  \pgfmathtruncatemacro{\lastDivider}{\pairTestBlockCount-1}
  \foreach \divider in {1,...,\lastDivider} {
    \draw[draw=pairTestInk, line width=0.45pt] ({8+\divider*\pairTestBlockWidth},6) -- ({8+\divider*\pairTestBlockWidth},11);
  }
  \node[pair test/label] at ({\pairTestFirstLeft+\pairTestBlockWidth/2},2.8) {block $i$};
  \node[pair test/label] at ({\pairTestSecondLeft+\pairTestBlockWidth/2},2.8) {block $j$};

  \foreach \gate/\offset in {1/-16,2/-8,3/0,4/16} {
    \node[pair test/gate] (pair-test-clause-\gate) at ({\pairTestCenter+\offset},20.3) {$\vee$};
  }
  \node[pair test/ellipsis] at ({\pairTestCenter+8},20.3) {$\cdots$};
  \node[pair test/gate] (pair-test-output) at (\pairTestCenter,31.7) {$\wedge$};
  \node[pair test/label, font=\fontsize{10.5}{12}\selectfont] at (\pairTestCenter,40.9) {$\mathbf{1}[\,w_i+s_i\not\equiv w_j+s_j\pmod{k}\,]$};
  \node[pair test/label, anchor=west, font=\fontsize{10}{12}\selectfont] at ({\pairTestCenter+17},28.5) {size $2^{O(n/k)}$};

  \begin{scope}[on background layer]
    \path[pair test/triangle] ({\pairTestCenter-27},16.6) -- ({\pairTestCenter+27},16.6) -- (\pairTestCenter,36) -- cycle;
    \foreach \fraction/\gate in {0.1/1,0.26/1,0.42/2,0.58/2,0.74/3,0.9/4} {
      \draw[pair test/wire] ({\pairTestFirstLeft+\fraction*\pairTestBlockWidth},11) -- (pair-test-clause-\gate.south);
    }
    \foreach \fraction/\gate in {0.1/1,0.26/2,0.42/3,0.58/3,0.74/4,0.9/4} {
      \draw[pair test/wire] ({\pairTestSecondLeft+\fraction*\pairTestBlockWidth},11) -- (pair-test-clause-\gate.south);
    }
    \foreach \gate in {1,...,4} {
      \draw[pair test/wire] (pair-test-clause-\gate.north) -- (pair-test-output.south);
    }
    \draw[pair test/output wire] (pair-test-output.north) -- (\pairTestCenter,37.5);
  \end{scope}
\end{tikzpicture}%
\endgroup%
\caption{Computing the function $\textbf{1}[w_i + s_i \not\equiv w_j + s_j \pmod k]$ as a small CNF, depending only on the input variables in blocks $i$ and $j$.}
\label{fig:pairwise-test}
\end{figure}

Finally, to certify that $|x| = t$, we take the $\AND$ of two independent tests (each with its own partition and fixed shifts) using $k = \ceil{\sqrt{n}}, \ceil{\sqrt{n}} + 1$. These moduli are coprime and their product is greater than $n$, so passing both tests certifies that $|x| = t$. The resulting test is a CNF of size $2^{O(\sqrt{n})}$.

\paragraph{Coverage with random tests.} Next, our goal is to show that without too many tests, we can cover every $x$ with weight $t$.

Start by considering the smaller test which certifies that $|x| \equiv t\pmod{k}$. We are interested in the probability that a valid $x$ (with weight $t$ modulo $k$) passes a test with shifts chosen uniformly at random subject to \Cref{eq:shift-constraint}. There are $k^{k-1}$ total choices of shifts, since the first $k - 1$ can be chosen arbitrarily, and the last is fixed by their sum. Of these, $k!$ certify any fixed valid $x$. To see this, choose the first $k-1$ shifts in order. Once $s_1, \dots, s_{i-1}$ have been chosen without collisions, there are $k - i + 1$ choices for $s_i$ such that $w_i + s_i$ does not collide with any $w_j + s_j$ for $j < i$. Since $x$ is valid, the constraint on the sum of the shifts then forces $w_k+s_k$ to take the one remaining residue. This gives $k\cdot(k-1)\cdots 2 = k!$ choices. Thus, the probability that a random test covers a valid $x$ is
$$\frac{k!}{k^{k-1}} \geq e^{-k}$$
by Stirling's approximation (or if you prefer, the Taylor series for $e^x$ at $x = k$). Since our test certifying $|x| = t$ uses two independent tests with $k = \ceil{\sqrt{n}}, \ceil{\sqrt{n}} + 1$, the success probability is at least 
$$p := e^{- (2\ceil{\sqrt{n}} + 1)} = 2^{-O(\sqrt{n})}.$$
It follows that if we choose $T$ tests independently and uniformly at random, then the expected number of valid inputs $x$ that do not pass any tests is at most 
$$2^n (1 - p)^T \leq 2^n \exp(-p T).$$
This expectation is less than 1 for $T \geq n/p = 2^{O(\sqrt{n})}$, so by the probabilistic method, there is a choice of at most $2^{O(\sqrt{n})}$ tests that cover every valid $x$. Taking the $\OR$ of these tests, we get a depth-3 circuit for the function $\textbf{1}[|x| = t]$ of size $2^{O(\sqrt{n})}$.
\end{proof}

We conclude the section with two brief remarks about the result and possible origins of the techniques. We thank Ryan Williams for pointing out these connections and helping us navigate the relevant literature. 

\begin{remark}\label{remark:SSETH}
\cite{GPPST24} propose an approach to faster $k$-SAT algorithms via a reduction to a certain \emph{local enumeration} problem. They show that an efficient enough algorithm for this problem would imply both the following:
\begin{enumerate}[label=(\arabic*)]
\item Expressing $\Majority$ as an $\OR$ of $k$-CNFs
requires $2^{\Omega(n \log k/k)}$ many CNFs. 
\item $k$-SAT can be solved in $2^{(1 - \Omega(\frac{\log k}{k}))n}$ randomized expected time, refuting Super Strong ETH (SSETH).
\end{enumerate}
Unfortunately, the construction above rules out (1). Switching to our notation, we can use our tests with any moduli $\{k, k+1\}$ such that $k \in [\ceil{\sqrt{n}}, n)$ (not just $\{\ceil{\sqrt{n}}, \ceil{\sqrt{n}} + 1\}$). For $k' = 2\ceil{n/k}$, we can compute $\Majority$ as an $\OR$ of $2^{O(k)} = 2^{O(n/k')}$ many $k'$-CNFs.

\end{remark}

\begin{remark}\label{remark:color-coding}
The critical idea of randomly hashing the block weights and testing for distinctness to save a $k!$ factor can be seen as an incarnation of the classic \emph{color-coding} trick of~\cite{AYZ1995}. To efficiently find a path on $k$ distinct vertices in a graph, they randomly color the vertices with $k$ colors and use dynamic programming to find a path with distinct colors. The success probability follows the same calculation that we have. In the last three decades, color-coding has become a widely used tool in parametrized algorithms since the same idea can be adapted to efficiently find many kinds of small structures. We refer the reader to~\citet[Chapter~8]{DF2013} and~\citet[Chapter~5]{CyganEtAl2015} for book chapters detailing these applications. Closer to home, a construction of \cite{Khasin1969} (as presented by \citet[Exercises~11.6 and 11.7]{Juk12} or \citet[Exercises~3.11 and 3.12]{Juk11}) gives monotone depth-3 formulas for the \emph{$t$-threshold function} ($\textbf{1}[|x| \geq t]$) of size $O(te^t n\log n)$, using an argument that can be seen as an even earlier incarnation of color-coding. The random-coloring argument at the heart of color-coding also appears as early as~\cite{ErdosKleitman1968}.

While we (frustratingly) cannot be sure where Astra's ideas came from, we suspect that the prevalence of this technique across the literature led to it being comfortably internalized by the model.

\end{remark}

\section{Depth \texorpdfstring{$d$}{d}}
\label{sec:depth-d}

Next, we turn to general depth. The key observation is that the function
$$\textbf{1}[w_i + s_i \not\equiv w_j + s_j \pmod k]$$
is almost symmetric. Since it only depends on the difference of weights between blocks $i$ and $j$, it is a symmetric function of block $i$'s input and the negation of block $j$'s input. This lets us recurse, making each test a small $\Pi_{d - 1}$ circuit rather than a CNF.

Inductively, each test now becomes a circuit of size $2^{O((n/k)^{1/(d-2)})}$, and the second part of the argument is essentially unchanged, still requiring $2^{O(k)}$ tests to cover every valid input. Balancing the exponents suggests setting $k \approx n^{1/(d-1)}$. The last difference is that using two moduli is no longer sufficient since we need their product to be greater than $n$, but taking $d - 1$ moduli does the trick.

\begin{maintheoremrestated}
For any constant $d\geq 2$, every symmetric function on $n$ variables has depth-$d$ circuits of size $2^{O(n^{1/(d-1)})}$.
\end{maintheoremrestated}

\begin{proof}

We will show by induction on $d$ that every symmetric function has $\Sigma_d$ and $\Pi_d$ circuits of size $2^{O(n^{1/(d-1)})}$. Since the negation of a symmetric function is also symmetric, by De Morgan's laws it suffices to construct just $\Sigma_d$ circuits (with an $\OR$ at the output layer). As before, it suffices to construct circuits for $\textbf{1}[|x| = t]$.

The base case $d = 2$ follows from the fact that any function can be represented as a CNF or DNF of size $2^{n}$ ($\Pi_2$ and $\Sigma_2$ circuits of size $2^{n} + 1$). Henceforth, fix $d \geq 3$ and assume that any symmetric function has  $\Sigma_{d-1}$ and $\Pi_{d - 1}$ circuits of size $2^{O(n^{1/(d-2)})}$. We will show how to construct a $\Sigma_d$ circuit for $\textbf{1}[|x| = t]$ of size $2^{O(n^{1/(d-1)})}$.

The approach follows the template of the depth-3 proof. We first construct \emph{tests}, represented as small $\Pi_{d-1}$ circuits, which certify that an input has weight $t$. Then we show, via the probabilistic method, that a small set of tests is sufficient to cover every input with weight $t$.

\paragraph{Anatomy of a test.} We start by constructing small tests that certify that $|x| \equiv t \pmod k$ for any $k$ (eventually set to be $\approx n^{1/(d-1)}$). 

Once again, each test splits the input $x$ into $k$ blocks of length at most $\ceil{n/k}$, and fixes shifts $s_1, \dots, s_k \in \{0, 1, \dots, k-1\}$ satisfying 
\begin{equation}\label{eq:constraint-d}
 \sum_{i=1}^k s_i \equiv \binom{k}{2} - t \pmod{k}.   
\end{equation}
Let $w_i$ be the weight of block $i$. The test passes if $w_1 + s_1, \dots, w_k + s_k$ are all distinct modulo $k$, which implies that the weight of $x$ is $t$ mod $k$. In other words, the test computes the function 
\begin{equation}\label{eq:test-d}
\bigwedge_{1 \leq i < j \leq k} \textbf{1}[w_i + s_i \not\equiv w_j + s_j \pmod k].
\end{equation} 
The key observation is that the constituent functions $\textbf{1}[w_i + s_i \not\equiv w_j + s_j \pmod k]$ depend on at most $2\ceil{n/k}$ variables, and are \emph{nearly} symmetric functions over these variables. Since these functions only depend on the difference $w_i - w_j$, they are symmetric if we treat their input as block $i$'s variables and the \emph{negations} of block $j$'s variables. Thus, the inductive hypothesis tells us that $\textbf{1}[w_i + s_i \not\equiv w_j + s_j \pmod k]$ can be expressed as a $\Pi_{d - 1}$ circuit of size
\begin{equation}\label{eq:size}
2^{O((n/k)^{1/(d-2)})}.
\end{equation}
Finally, to get a test for $|x| = t$, we combine separate tests (each with its own partition and shifts) over several moduli $k_1, k_2, \dots, k_{d - 1}$. If the moduli are all pairwise coprime and greater than $n^{1/(d-1)}$, their product is greater than $n$, and so
$$\bigwedge_{i = 1}^{d-1} \textbf{1}[|x| \equiv t \pmod{k_i}] = \textbf{1}[|x| = t].$$
For the sake of the next step, which requires a reasonably large probability that a fixed valid input passes a random test, we need moduli that are not too much larger than $n^{1/(d-1)}$. We can arrange this by making $k_i$ the smallest power of the $i$th prime $p_i$ that is greater than $n^{1/(d-1)}$, in which case $k_i \leq p_i n^{1/(d-1)} = O(n^{1/(d-1)})$. For any such $k$, the bound given by \Cref{eq:size} is $2^{O(n^{1/(d-1)})}$.  

Ultimately, the test for $|x| = t$ is the $\AND$ of polynomially many circuits of this size. We can collapse their output $\AND$ gates, so the overall test is implemented as a $\Pi_{d - 1}$ circuit of size $2^{O(n^{1/(d-1)})}$.

\paragraph{Coverage with random tests.} Next, take a fixed input $x$ with weight $t$, and suppose that we choose a random test for $|x| = t$ by choosing shifts for each modulus uniformly subject to the constraint in \Cref{eq:constraint-d}, independently across moduli. By the same argument as in the depth-3 case, a given valid $x$ passes the test for a modulus $k$ with probability at least $e^{-k}$, and so it passes the tests for every modulus with probability at least 
$$p := e^{-\sum_i k_i} = 2^{-O(n^{1/(d-1)})}.$$
With $T$ independently chosen random tests, the expected number of valid inputs $x$ that do not pass any tests is at most
$$2^n(1 - p)^T \le 2^n \exp(-p T)$$
which is less than $1$ for  $T \geq n/p = 2^{O(n^{1/(d-1)})}$. Thus, we can compute $\textbf{1}[|x| = t]$ with the $\OR$ of $2^{O(n^{1/(d-1)})}$ tests, each of which is a $\Pi_{d - 1}$ circuit of size $2^{O(n^{1/(d-1)})}$, giving a $\Sigma_d$ circuit of size $2^{O(n^{1/(d-1)})}$. 
\end{proof}

\section{Discussion}
\label{sec:discussion}

The results in this paper and the methods used to obtain them can be taken as cause for pessimism or optimism. 

On the one hand, now that $\Majority$ is no longer the simplest function whose circuit complexity is poorly understood, it seems less likely that circuit lower bound breakthroughs will come from studying simple functions. While $\Majority$ served as a tractable sandbox for developing new circuit lower bound techniques, the next sandbox may involve pricklier toys that are less amenable to elegant theory. 

This view makes the road to understanding circuit complexity seem long and uncertain, but \Cref{thm:main} can also be interpreted as showing that we were actually further along than we thought. Our inability to improve circuit lower bounds for $\Majority$ seemed like an indicator of how limited our techniques were, but it turns out that the switching lemma gave the right bound all along. More broadly, we are hopeful that new, powerful theorem proving tools can enable us to revisit research programs that stalled on thorny old open problems, and continue to find simple resolutions that were hiding just beneath the surface.

\paragraph{Acknowledgments.} Thanks to Li-Yang Tan for introducing us to the problem, and for many stimulating conversations about it. We also thank Ryan Williams for alerting us to the connection to SSETH from by \cite{GPPST24}, and the importance of the color-coding technique of \cite{AYZ1995} in the FPT algorithms literature (discussed in \Cref{remark:SSETH,remark:color-coding} respectively).

\nocite{LRT22}
\bibliographystyle{plainnat3}
\bibliography{references}

\end{document}